\documentclass[preprint,authoryear,12pt]{elsarticle}
\usepackage{graphicx}
\usepackage{pstricks,pst-plot,pst-node,pst-text,pst-3d, pst-gr3d}
\usepackage{amsfonts}
\usepackage{amssymb}
\usepackage{amsmath}
\usepackage{booktabs}
\usepackage{hyperref}
\usepackage{tikz}
\usepackage{mathabx}
\usepackage{eurosym}
\usepackage[singlespacing]{setspace}
\usepackage{geometry}
\usepackage{comment}
\usepackage{caption}
\usepackage{subcaption}

\newtheorem{axiom}{Axiom}

\newtheorem{proposition}{Proposition}

\newtheorem{example}{Example}

\newtheorem{corollary}{Corollary}

\newenvironment{proof}[1][Proof]{\textbf{#1.} }{\ \rule{0.5em}{0.5em}}
\usetikzlibrary{arrows,shapes}
\usetikzlibrary{automata}
\hypersetup{    
    colorlinks = {true},
    linktocpage = {true},
    plainpages = {false},
    linkcolor = {Blue},
    citecolor = {Blue},
    urlcolor = {Red},
    pdfstartview = {Fit},
    pdfview = {XYZ null null null}
  }
\journal{unknown}

\begin{document}

\begin{frontmatter}

\title {\bf{The value and credits of $n$-authors publications}} 

\author{Lutz Bornmann$^{\dagger}$ and Ant\'{o}nio Os\'{o}rio$^{\ddagger}$}

\address {$^{\dagger}$Max Planck Society Munich MPG (bornmann@gv.mpg.de) 
\linebreak $^{\ddagger}$Universitat Rovira i Virgili (Dept. of Economics) and CREIP (antonio.osoriodacosta@urv.cat).}


\begin{abstract}

Collaboration among researchers is becoming increasingly common, which raises a large number of scientometrics questions for which there is not a clear and generally accepted answer. For instance, what value should be given to a two-author or three-author publication with respect to a single-author publication? This paper uses axiomatic analysis and proposes a practical method to compute the expected value of an $n$-authors publication that takes into consideration the added value induced by collaboration in contexts in which there is no prior or ex-ante information about the publication's potential merits or scientific impact. The only information required is the number of authors. We compared the obtained theoretical values with the empirical values based on a large dataset from the Web of Science database. We found that the theoretical values are very close to the empirical values for some disciplines, but not for all. This observation provides support in favor of the method proposed in this paper. We expect that our findings can help researchers and decision-makers to choose more effective and fair counting methods that take into account the benefits of collaboration.

\end{abstract}


\end{frontmatter}

\section{\textbf{Introduction}}

%
%
%
%
%
%
%

Collaboration among researchers is becoming increasingly common, which may
reflect the increasing complexity and interdisciplinary content of research (%
\citealp{gazni2012mapping}; \citealp{katz1997research}; %
\citealp{lariviere2015team}; \citealp{persson2004inflationary}; %
\citealp{wuchty2007increasing}). In this context, the more researchers are
involved in a project, the more difficult it is for third parties (e.g., a
reviewers' panel or an evaluation committee) to quantify or observe (even
imperfectly) the contribution of each researcher. The problem is so severe
that some authors in the literature propose that publications should
unambiguously list the specific contribution of each author (%
\citealp{cronin2001hyperauthorship}; \citealp{hu2009loads}; %
\citealp{tscharntke2007author}).

The increasing collaboration raises a large number of scientometrics
questions for which there is not a clear and generally accepted answer. For
instance, what value should be given to a two-author or a three-author
publication with respect to a single-author publication? Who should be
ranked first, an individual with a single-author publication or an
individual with two three-author publications? The answer to these questions
is crucial because academics and researchers (as well as their institutions)
are ranked, rewarded, financed and promoted according to the quantity and
quality of their publications. Consequently, we need to develop adequate
counting methods that are sufficiently flexible and that can benefit from
generalized support.

\bigskip

The problem of how to count publications with several authors has been
discussed previously in the literature (\citealp{egghe2000methods}; %
\citealp{lindsey1980production}; \citealp{price1981multiple}; among others).
The most common solution is to allocate the credits proportionally to the
number of authors (fractional counting). This practice is incomplete because
it ignores the potential synergies that result from collaborations, and
consequently underestimates the credits of each author and the overall value
of the publication. Publications with several authors tend to have more
citations than single-author publications (\citealp{hsu2011correlation}; %
\citealp{onodera2015factors}; among others).\footnote{%
For instance, \cite{hsu2011correlation} approximated the relation between
the number of authors\ and the number of citations by the expression $%
citations_{n}=(n/5)^{1/3},$ where $n$ denotes the number of authors.}
Therefore, the value of a publication should increase with the number of
authors. This issue has been consistently ignored in the literature.

Another common solution is to allocate the full credit of a publication to
every author (full counting). This practice overestimates the credits of
each author and the value of the publication---each author is treated as a
single author, which leads to serious distortions in comparing individuals
with different co-authorship patterns. Moreover, it creates incentives to
the addition of \textquotedblleft ghost" co-authors, which is not desirable.
Thus, as pointed out by \cite{hirsch2010index}, we need adequate counting
methods that take into consideration the number of authors involved in a
publication. This issue is the objective of this paper.

The methods discussed so far are useful when all authors are equally
important, as for example, when authors are ordered alphabetically, which is
common in mathematics, economics, finance, and high energy physics (%
\citealp{frandsen2010name}; \citealp{hu2009loads}; %
\citealp{maruvsic2011systematic}; \citealp{waltman2012empirical}). However,
in other scientific fields, authors are ordered in accordance with their
contribution to the publication, with the first author typically being
regarded as the most important. In this context, several approaches have
been proposed in the literature (\citealp{waltman2016review}). One
possibility is to allocate the full credit to the first author and/or the
corresponding author, which can be more than one, and no credit to the other
authors (\citealp{egghe2000methods}; \citealp{gauffriau2007publication}; %
\citealp{huang2011counting}; \citealp{lange2001citation}; %
\citealp{van1997fractional}).

However, in general, all listed authors have contributed to the publication
and for that reason should receive some credit, with the most credits being
given to the first author, followed by the second author, and so on. Several
distributions of credits have been proposed in the literature. These include
the arithmetic counting method (\citealp{van1997fractional}), in which
credits are linearly distributed in decreasing order among the authors, the
geometric counting method (\citealp{egghe2000methods}), in which each author
always gets twice the credits of the following author, the harmonic counting
method (\citealp{hagen2008harmonic}; \citealp{sekercioglu2008quantifying}),
in which the $i$-th ranked author receives $1/i$ of the credit received by
the first author, and the axiomatic counting method (%
\citealp{stallings2013determining}), which is conceptually the most similar
to that proposed in the present paper. Other counting methods and procedures
have been proposed (\citealp{abramo2013importance}; %
\citealp{assimakis2010new}; \citealp{kim2014network}; \citealp{liu2012fairly}%
; \citealp{lukovits1995correct}; \citealp{trueba2004robust}). All these
methods are based on some intuitively correct argument, which makes it
difficult to compare them or to claim that one method is superior to the
others (\citealp{kim2015rethinking}; \citealp{xu2016author}).

\bigskip

This paper proposes a practical and simple method to compute the expected
value of an $n$-authors publication that takes into consideration the
potential added value induced by collaboration in contexts in which there is
no prior or ex-ante information about the publication's potential merits or
scientific impact. The only information required is the number of authors.
The method is neutral to the identity and affiliation of the authors, and it
is flexible enough to accommodate information such as the order of authors,
the Journal Impact Factor, the number of citations, or to be used in
connection with other scientometrics indicators, like the \textit{h}-index (%
\citealp{hirsch2005index}) and its several variations and extensions that
have been proposed in the literature (\citealp{bornmann2011multilevel}).
These aspects make the proposed method extremely practical and useful to
deal with real life situations in which consensus is difficult, and
distinguishes the present paper from the existing literature.

\bigskip

In this context, we apply a set of principles or axioms that we consider
fundamental to determine the expected value of an $n$-authors publication.%
\footnote{%
Axiomatic approaches have been frequently used to study problems in
scientometrics. For instance, in the axiomatic characterization of
bibliometric impact indicators, as the \textit{h}-index (%
\citealp{deineko2009new}; \citealp{kongo2014alternative}; %
\citealp{miroiu2013axiomatizing}; %
\citealp{quesada2009monotonicity,quesada2010more,quesada2011further}; %
\citealp{woeginger2008axiomatic,woeginger2008symmetry}) or some of its
variants (\citealp{adachi2015further}; \citealp{quesada2011axiomatics}; %
\citealp{woeginger2008axiomatic0,woeginger2009generalizations}), as well as
the Euclidian index (\citealp{perry2016count}), and in the axiomatic
characterization of rankings of authors and journals derived from
bibliometric indicators (%
\citealp{bouyssou2010consistent,bouyssou2011bibliometric,bouyssou2014axiomatic,bouyssou2016ranking}%
; \citealp{marchant2009axiomatic}).} First, we present two basic and
generally accepted inequalities regarding the expected value of $n$-authors
publications. These inequalities will help us to build our argument.
Subsequently, we require that the value of a publication is equal to the
aggregated sum of all authors' efforts and that successful collaboration
must satisfy some minimum amount of aggregate effort. In addition, there is
an upper bound on the maximum effort provided by each co-author. In this
context, each author expends its effort in the collaboration that returns
the largest amount of credits and has no incentives to expend effort in any
other collaboration. Lastly, in order to obtain analytical results we
consider that every feasible effort is equally likely.

The result is a unique expression for the expected value of $n$-authors
publications that increases (at a decreasing rate) with the number of
authors. Some properties of the proposed method are discussed in connection
with the existing literature.

\bigskip

Finally, using citation impact data, we contrasted the expected values with
empirical bibliometric data from the Web of Science (WoS, Clarivate
Analytics). We found that the theoretically expected values are very close
to the empirical values. This observation provides strong support in favor
of this method of calculating the expected value of $n$-authors publications
described in this paper. Nonetheless, these results also make explicit that
this method should not be taken as a universal solution to all disciplines.

\bigskip

This paper is organized as follows: Section \ref{sec_equal} presents the
main axioms and provides explanatory information, Section \ref{sec_main}
presents the theoretical results, Section \ref{sec_compare}\ compares the
obtained theoretical results with bibliometric data, and Section \ref%
{sec_concl} discusses the results of this study.

\section{\textbf{The axioms of an }$n$\textbf{-authors publication}\label%
{sec_equal}}

Suppose that $n=1,...,\infty $ authors collaborate on an academic or
scientific project. Suppose that each of these authors is equally important,
i.e., they are expected to provide the same effort and obtain the same
credits, which does not mean that they will necessarily do it.

Let $v_{n}$ denote the associated $n$-authors publication value and $%
\overline{v}_{n}$ denote the associated $n$-authors publication\ expected
value. We distinguish between the publication value and the publication
expected value. The former is unique to each collaboration and unknown to
third parties, while the latter is an estimation of this value. The
objective of this paper is to approximate the latter value.

Let $\overline{c}_{n}$ denote the\ credits awarded to each author, which in
the case that all authors are equally important corresponds to the
publication expected value divided by the number of authors, i.e., $%
\overline{c}_{n}=\overline{v}_{n}/n.$ Let $e_{in}$ denote the effort or
contribution of author $i=1,...,n$ in the $n$-authors publication or
collaboration. In our context, $e_{in}$ captures simultaneously the
quantitative and the qualitative dimensions of effort.

\bigskip

In what follows, we present and discuss two basic inequalities that should
be satisfied by the expected value of an $n$-authors publication.
Subsequently, we present a set of axioms that relate effort, value and the
publication expected value, and that characterize the approach in this paper.

\subsection{\textbf{Basic inequalities and discussion}}

We start by noting that the addition of more authors should increase the
expected value of a publication (\citealp{hsu2011correlation}; %
\citealp{onodera2015factors}; among others), or at least not decrease it.
The interaction of authors with potentially different experiences and
knowledge generates positive synergies and the cross-fertilization of ideas.
In other words, the following inequality should be always satisfied:%
\begin{equation}
\overline{v}_{1}\leq \overline{v}_{2}\leq ...\leq \overline{v}_{n}\text{ \
for \ }n=1,...,\infty .  \label{chain_value}
\end{equation}

The expected value of a publication with $n$ authors must have at least the
same value as a publication with $n-1$ authors, and so on. The aggregate
effort of $n$ authors should result in something quantitatively and
qualitatively better than the aggregate effort of $n-1$ authors. The
question is how much better?

\bigskip

In this context, if all authors are equally important, i.e., they provide
the same expected effort, the publication average value (which is also the
credits awarded to each author if all authors are equally important) must
decrease with the number of authors. Otherwise, there would be an incentive
to add "ghost" co-authors to the collaboration because the publication
average value would increase. In other words, the following inequality
should be always satisfied:%
\begin{equation}
\frac{\overline{v}_{1}}{1}\geq \frac{\overline{v}_{2}}{2}\geq ...\geq \frac{%
\overline{v}_{n}}{n}\text{ \ for \ }n=1,...,\infty .  \label{chain_average}
\end{equation}

Therefore, in the case that all authors are equally important, the\ credits
awarded to each author and the publication average value are the same, i.e., 
$\overline{c}_{n}=\overline{v}_{n}/n$ for $n=1,...,\infty .$

In our context, inequalities (\ref{chain_value}) and (\ref{chain_average})
are general and intuitive, and they could have been written in the form of
axioms.\footnote{%
For instance, inequality (\ref{chain_value}) appears in \cite%
{stallings2013determining} as an axiom. Monotonicity axioms are common in
studies characterizing bibliometric impact indicators. For instance, it is
common to impose that the number of publications and/or citations should not
lower the value of the indicator (\citealp{adachi2015further}; %
\citealp{deineko2009new}; \citealp{kongo2014alternative}; %
\citealp{quesada2009monotonicity,quesada2011axiomatics,quesada2011further}; %
\citealp{woeginger2008axiomatic0,woeginger2008axiomatic,woeginger2008symmetry,woeginger2009generalizations}%
).} However, in order to avoid dependence issues with the axioms presented
below, we will not do it.

Together, inequalities (\ref{chain_value}) and (\ref{chain_average}) imply
the following bounds on the expected value of an $n$-authors publication: 
\begin{equation}
\overline{v}_{n-1}\leq \overline{v}_{n}\leq \frac{n}{n-1}\overline{v}_{n-1}%
\text{ \ for \ }n=2,...,\infty .  \label{basic_ineq}
\end{equation}%
In other words, the expected value of an $n$-authors publication should be
in the interval defined by inequality (\ref{basic_ineq}).\ 

\bigskip

In this context, the most extreme upper bound in this interval is obtained
when $\overline{v}_{n}=\frac{n}{n-1}\overline{v}_{n-1}$ for $n=2,...,\infty
. $ Consequently, we obtain recursively that $\overline{v}_{n}=n\overline{v}%
_{1}$ and $\overline{c}_{n}=\overline{v}_{1}$ for $n=1,...,\infty .$\ The
rule that results from this upper bound is often found in practice and is
based on awarding the full value of a single-author publication to each of
the $n$\ authors (full counting). This practice is excessively generous
because the value of a single-author publication is multiplied by the number
of authors.\ 

Similarly, the most extreme lower bound in this interval is obtained when $%
\overline{v}_{n}=\overline{v}_{n-1}$ for $n=2,...,\infty .$\ In this case,
we obtain recursively that $\overline{v}_{n}=\overline{v}_{1}$ and $%
\overline{c}_{n}=\overline{v}_{1}/n$ for $n=1,...,\infty .$\ The rule that
results from this lower bound corresponds to another commonly used practice,
which consists in dividing the value of a single-author publication equally
among the $n$ authors (fractional counting). This practice is more similar
to the one in this paper, but ignores the added value resulting from
collaboration and, consequently, underestimates the value of a publication.

\subsection{\textbf{Axioms}}

The bounds found in inequality (\ref{basic_ineq}) are frequently seen as
insufficient. Consequently, they are unable to provide meaningful and
acceptable predictions about the expected value of an $n$-authors
collaboration. In what follows, we present a set of axioms, which will allow
us to obtain an analytical expression for the expected value of an $n$%
-authors publication. Then, as a corollary to this main result, we will see
that the satisfaction of those axioms implies the satisfaction of inequality
(\ref{basic_ineq}).

\subsubsection{\textbf{Axioms about aggregated effort\label{subsect_aggreg}}}

We start by noting that the value of an $n$-authors publication ($v_{n}$)
must take into consideration the effort of each of the $n$ authors ($e_{in}$%
), i.e., the value of an $n$-authors publication must be equal to the
aggregated sum of the efforts of the $n$ authors participating on it.

\begin{axiom}[aggregate collaboration function]
\label{colab_funct}$v_{n}=e_{1n}+e_{2n}+...+e_{nn}$\ for $n=1,...,\infty .$
\end{axiom}

Collaboration in Axiom \ref{colab_funct} is expressed in the most neutral
and simplest way: as the sum of the $n$ authors' contributions. However, we
could have considered other functional forms. For instance, we could have
considered some terms that would explicitly capture synergies and positive
interaction effects resulting from collaboration between individuals with
potentially different experiences and knowledge. Similarly, we could have
considered some terms that would explicitly capture the increasing
coordination difficulty and the possibility of free-riding, which frequently
undermine the processes of collaboration. However, we have no theory to
support any such particular functional form. Consequently, the inclusion of
such terms would be pretty much ad-hoc and questionable, and would
inevitably influence the results in one way or another. For that reason,
collaboration in Axiom \ref{colab_funct} is expressed in the most neutral
and simplest way.

\bigskip

The set of possible effort profiles $(e_{1n},e_{2n},...,e_{nn})$ that can be
considered in Axiom \ref{colab_funct} is uncountable and very diverse.
However, not all these effort profiles can be accepted as valid. Some effort
profiles have associated such low levels of aggregate effort that successful
collaboration is impossible. Consequently, not all collaborations will
necessarily lead to a publication.

In our context, a publication is only possible if the aggregate effort is
above some threshold. The following axiom establishes the minimum amount of
aggregate effort necessary for an $n$-authors' collaboration to result in a
publication.

\begin{axiom}[minimum aggregate effort]
\label{aggregated_eff}Successful collaboration must satisfy: $%
e_{1n}+e_{2n}+...+e_{nn}\geq \overline{v}_{n-1}$ for $n=1,...,\infty .$
\end{axiom}

An $n$-authors publication is only possible if the aggregate effort is above
the expected value of an $n-1$-authors publication. In other words, the
value of an $n$-authors publication must exceed the expected value of an $%
n-1 $-authors publication, and the value of an $n-1$-authors publication
must exceed the expected value of an $n-2$-authors publication, and so on.
The idea is that the consideration of additional co-authors must lead to an
increase in value above the expected value of the nearest collaboration with
a lower number of co-authors. Otherwise, individuals would have no
incentives to consider additional co-authors.

\subsubsection{\textbf{Axioms about individual effort}}

The main difficulty regarding the individual effort provided by each author (%
$e_{in}$) is the fact that it is not observable by third parties (e.g., a
reviewers' panel or an evaluation committee), even imperfectly---and
sometimes not even observable by the other co-authors.

In this context, in order to obtain analytical results and pointwise
predictions about the expected value $\overline{v}_{n},$ the uniform
distribution seems to be the most focal and neutral assumption, because it
has implicitly an impartial and equal treatment of all possible efforts.
This idea becomes even more natural if we have no theory to support any
other distribution (the same axiom appears in \cite{stallings2013determining}%
).

\begin{axiom}[uniform distribution of effort]
\label{unif_dist}$e_{in}$ is independent and uniformly distributed for $%
i=1,...,n$ and $n=1,...,\infty .$
\end{axiom}

Consequently, in order to obtain the expected value of $\overline{v}_{n}$,
we must build expectations about the effort provided by each author. In this
context, Axioms \ref{colab_funct} and \ref{unif_dist} imply that $\overline{v%
}_{n}=E(e_{1n})+E(e_{2n})+...+E(e_{nn})$\ for $n=1,...,\infty .$

Moreover, in order to allocate the publication credits in the most fair and
objective way, the credits awarded to each author in an $n$-authors
publication must be equal to the expected effort provided.\footnote{%
Equity and fairness require an adequate balance between reward and effort.
In this context, the sense of fairness depends on the individual comparison
between their own balance and the balance of the others, with whom the
individual deems to be relevant references (\citealp{adams1963}). In our
context, the effort used as reference is the effort provided by the other
co-authors, while the reward used as reference is the credits awarded to the
other co-authors.} In this context, in the case that all authors are equally
important, they are expected to provide the same level of effort and obtain
the same credits, i.e., $E(e_{1n})=...=E(e_{nn})=\overline{e}_{n}=\overline{c%
}_{n}$ and $\overline{e}_{n}=\overline{c}_{n}=\overline{v}_{n}/n$ for $%
n=1,...,\infty $.

\bigskip

In Section \ref{subsect_aggreg}, the joint consideration of Axioms \ref%
{colab_funct} and \ref{aggregated_eff} establishes a floating lower bound on
the individual effort, but in order to have a well-defined problem, we also
need an upper bound on the individual effort.

In this context, we consider that an author is someone with a certain
maximum amount of effort to spend in a research collaboration. Therefore,
before taking part in a research collaboration, a rational author takes into
consideration the expected effort required and the credits obtained in each
collaboration, in order to maximize the amount of credits obtained.
Consequently, if the author $i$ participates in an $n$-authors collaboration
that means that the maximum amount of available effort that author $i$ has
to spend in that collaboration is lower than the expected effort required to
participate in an $n-1$-authors collaboration, i.e., $e_{in}\leq \overline{e}%
_{n-1}$ for $i=1,...,n$ and $n=1,...,\infty .$ This behavior is rational
because the credits tend to decrease with the number of authors (see
inequality (\ref{chain_average})). Therefore, author $i$ chooses the
research collaboration with the lowest possible number of authors, but in
which its effort is expected to be enough. Otherwise, the $n$-authors
collaboration would not be stable, because author $i$ would have an
incentive to move its effort to a collaboration with a lower number of
authors, but that would return more credits.

\begin{axiom}[maximum individual effort]
\label{incent_compat}$e_{in}\in \lbrack 0,\overline{e}_{n-1}]$\ for $%
i=1,...,n$ and $n=2,...,\infty .$
\end{axiom}

This axiom implies that as the number of authors increases, the maximum
amount of individual effort spent by each author decreases, which is in line
with inequality (\ref{chain_average}), and the connection between expected
effort and reward made after Axiom \ref{unif_dist}, i.e., $\overline{e}_{n}=%
\overline{c}_{n}=\overline{v}_{n}/n\leq \overline{e}_{n-1}=\overline{c}%
_{n-1}=\overline{v}_{n-1}/(n-1)$ for $n=2,...,\infty .$ Otherwise, each
author would prefer to spend that same amount of effort in a publication
with fewer co-authors that would award more credits. In this context, Axiom %
\ref{incent_compat} can also be seen as an incentive compatible and a
collaboration stability condition.

\bigskip

The following example illustrates some of the arguments that have motivated
the axioms presented in this paper.

\begin{example}
\label{example}Suppose that $n=3$ and $\overline{v}_{2}=4/3$ which implies
that $\overline{e}_{2}=\overline{c}_{2}=2/3.$ Therefore, according to Axioms %
\ref{unif_dist} and \ref{incent_compat}, the individual effort can take any
value in the interval $[0,2/3]$ with equal probability, but conditional on
satisfying the minimum aggregate effort condition $e_{13}+e_{23}+e_{33}\geq 
\overline{v}_{2}=4/3$ of Axiom \ref{aggregated_eff}.\ In this context, the
effort profiles $(e_{13},e_{23},e_{33})=(0.65,0.60,0.10)$ and $%
(e_{13},e_{23},e_{33})=(0.45,0.45,0.45),$ among many other effort profiles,
satisfy the minimum aggregate effort condition of Axiom \ref{aggregated_eff}
and for that reason are valid effort profiles in the computation of the
publication value of Axiom \ref{colab_funct}. This example considers only
two effort profiles, but the set of feasible effort profiles is uncountable
and very heterogeneous. It is the consideration of all these feasible effort
profiles that allows us to obtain a unique prediction of the $n$-authors
publication expected value. On the other hand, the effort profiles $%
(e_{13},e_{23},e_{33})=(0.65,0.20,0.10)$ and $%
(e_{13},e_{23},e_{33})=(0.40,0.40,0.40),$ among many other effort profiles,
do not satisfy the minimum aggregate effort condition of Axiom \ref%
{aggregated_eff} and for that reason are not valid because they would result
in collaboration failures.
\end{example}

\section{\textbf{The value of an }$n$\textbf{-authors publication\label%
{sec_main}}}

In this section, in order to compute the expected value of an $n$-authors
publication, we consider all possible effort configurations that
simultaneously satisfy Axioms \ref{colab_funct}-\ref{incent_compat}. These\
axioms determine how and which effort profiles should be considered in the
computation of the expected value of an $n$-authors publication. The
following proposition presents the main result in this paper.

\begin{proposition}
\label{prop_equal_credits}Suppose that Axioms \ref{colab_funct}-\ref%
{incent_compat} are satisfied. Then, the expected value of an $n$-authors
publication is uniquely given by:%
\begin{equation}
\overline{v}_{n}=\frac{2n}{n+1}\overline{v}_{1},  \label{exp_equal_credits}
\end{equation}%
for $n=1,...,\infty ,$ where $\overline{v}_{1}$ is the value of a
single-author publication, and the amount of credits awarded to each author
when all authors are equally important is uniquely given by $\overline{c}%
_{n}=\overline{v}_{n}/n$ for $i=1,...,n$ and $n=1,...,\infty .$
\end{proposition}

A few comments regarding the meaning and interpretation of the result
obtained are as follows.

First, by construction, expression (\ref{exp_equal_credits}) is uniquely
characterized by Axioms \ref{colab_funct}-\ref{incent_compat}.

Second, the expected value and the credits of the $n$-authors publications
are normalized with respect to the value of the single-author publication $%
\overline{v}_{1}.$ Once we know or attribute a value to $\overline{v}_{1},$
we can compute the value of $\overline{v}_{n}$ and $\overline{c}_{n}.$ In
this context, the value of the single-author publication is the unit of
measure.

Third, the results obtained must be interpreted in expected terms. Clearly,
in reality, not all $n$-authors publications are worth the same. Some
publications may be worth $n$ times more than a single-author publication
because of strong synergies and other interaction effects. In the other
extreme, some other publications may be worth as much as a single-author
publication because of free-riding and coordination problems. In between, we
may have single-authored publication that are worth more than multi-authored
publications, and so on.\footnote{%
This type of uncertainty and lack of information is the root of the problem
of computing each author's contribution, and the reason why several authors
in the literature suggest that publications should unambiguously list the
contribution of each author (\citealp{cronin2001hyperauthorship}; %
\citealp{hu2009loads}; \citealp{tscharntke2007author}).}

\begin{figure}[h]
\begin{center}
\includegraphics[width=.7\textwidth]{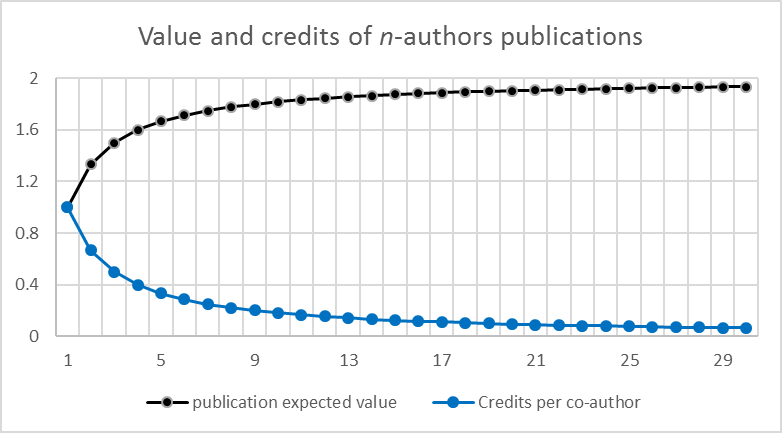}
\end{center}
\caption{{\protect\footnotesize Equaly important authors - credits awarded
to each author and the expected value of the $n$-authors publication for $%
n=1,\ldots ,30$ (unit of measure $\overline{v}_{1}=\overline{c}_{1}=1$).}}
\label{fig_proposition}
\end{figure}

Fourth, note that $\overline{v}_{n}$ converges to $2\overline{v}_{1}$ as $%
n\rightarrow \infty $. In other words, regardless of the number of
co-authors, the expected value of an $n$-authors publication cannot be
larger than twice the expected value of a single-author publication (see
Figure \ref{fig_proposition}). As we will see below, this observation may
not be empirically supported for all disciplines. However, we point out that
publications with large number of co-authors tend to be specific and very
particular (\citealp{cronin2001hyperauthorship}). For instance, \cite%
{xu2016author} point out that most credit allocation methods fail when there
are many co-authors and that we need specific methods to deal with
hyper-authorship, which is understood as ten or more co-authors (%
\citealp{liu2012fairly}; \citealp{tscharntke2007author}).

Fifth, as mentioned before, we do not need to explicitly impose inequalities
(\ref{chain_value}) and (\ref{chain_average}) in order for expression (\ref%
{exp_equal_credits}) in Proposition \ref{prop_equal_credits} to satisfy the
bounds implied by those inequalities (see inequality (\ref{basic_ineq})).
The following result formalizes this observation.

\begin{corollary}
\label{coro_satisfies}Expression (\ref{exp_equal_credits}) satisfies
inequalities (\ref{chain_value}) and (\ref{chain_average}).
\end{corollary}

Table \ref{tab_results1} shows the obtained expected value and the credits
awarded to each author in the case that all authors are equally important.
In the case of single-author publications, the expected value and the
credits awarded have the same value, which is normalized to the unit, i.e., $%
\overline{v}_{1}=\overline{c}_{1}=1.$

\begin{table}[tbp] \centering%
\begin{tabular}{cccccccccc}
\hline
& ${\small n=1}$ & ${\small n=2}$ & ${\small n=3}$ & ${\small n=4}$ & $%
{\small n=5}$ & ${\small n=6}$ & ${\small n=7}$ & ${\small n=8}$ & ${\small %
n=9}$ \\ \hline
$\overline{{\small c}}_{n}$ & ${\small 1.000}$ & ${\small 0.66}\overline{%
{\small 6}}$ & ${\small 0.500}$ & ${\small 0.400}$ & ${\small 0.33}\overline{%
{\small 3}}$ & ${\small 0.286}$ & ${\small 0.250}$ & ${\small 0.222}$ & $%
{\small 0.200}$ \\ 
&  &  &  &  &  &  &  &  &  \\ 
$\overline{v}_{n}$ & ${\small 1.000}$ & ${\small 1.33}\overline{{\small 3}}$
& ${\small 1.50}$ & ${\small 1.60}$ & ${\small 1.66}\overline{{\small 6}}$ & 
${\small 1.714}$ & ${\small 1.750}$ & ${\small 1.778}$ & ${\small 1.800}$ \\ 
\hline
\end{tabular}%
\caption{\footnotesize Equaly important authors - credits awarded to each author and the
expected value of the $n$-authors publication (unit of measure $\overline{v}_{1}=\overline{c}_{1}=1$).}%
\label{tab_results1}%
\end{table}%

Note that as the number of authors and thus the expected value of the
publication increases, the credits awarded to each author do not. The
expected value increases monotonically with the number of authors, but at a
decreasing rate (i.e., $\overline{v}_{n}$ is concave in the number of
authors), which implies that the marginal contribution of a new author
becomes less and less significant. Consequently, the credits awarded to each
author decrease.

\bigskip

The expected value and credits of an $n$-authors publication in Proposition %
\ref{prop_equal_credits} have some interesting properties. For instance, if
all authors are equally important, then our results state that in scientific
terms \textquotedblleft two $2$-authors publications are worth more credits
than one single-author publication (i.e., ${\small 0.66}\overline{{\small 6}}%
+{\small 0.66}\overline{{\small 6}}{\small >1}$)", and \textquotedblleft two 
$3$-authors publications are worth the same credits as one single-author
publication (i.e., ${\small 0.500}+{\small 0.500=1}$)", and so on. Note also
that in our context, the overall expected value of a $2$- and a $3$-authors
publication are worth $33\%$ and $50\%$ more than a single-author
publication, respectively. The exact relations are shown in Table \ref%
{tab_results1}.

\subsection{\textbf{The case of }$n$\textbf{\ ordered authors}\label%
{sec_diff}}

Finally, in the case that authors' importance is determined by the order in
which their names appear, we can also apply the publication expected value
obtained in Proposition \ref{prop_equal_credits}, but with the credits
distributed among the ordered co-authors according to one of the methods
proposed in the literature (\citealp{abramo2013importance}; %
\citealp{assimakis2010new}; \citealp{egghe2000methods}; %
\citealp{hagen2008harmonic}; \citealp{kim2014network}; %
\citealp{liu2012fairly}; \citealp{lukovits1995correct}; %
\citealp{sekercioglu2008quantifying}; \citealp{stallings2013determining}; %
\citealp{trueba2004robust}; \citealp{van1997fractional}).

\begin{table}[tbp] \centering%
\begin{tabular}{ccccccc}
\hline
& ${\small n=1}$ & ${\small n=2}$ & ${\small n=3}$ & ${\small n=4}$ & $%
{\small n=5}$ & ${\small n=6}$ \\ \hline
$\overline{{\small c}}_{1n}$ & ${\small 1.000(100)}$ & ${\small 1.000(75.0)}$
& ${\small 0.91}\overline{{\small 6}}{\small (61.}\overline{{\small 1}}%
{\small )}$ & ${\small 0.83}\overline{{\small 3}}{\small (52.0)}$ & ${\small %
0.76}\overline{{\small 1}}{\small (45.}\overline{{\small 6}}{\small )}$ & $%
{\small 0.700(40.8)}$ \\ 
$\overline{{\small c}}_{2n}$ &  & ${\small 0.33}\overline{{\small 3}}{\small %
(25.0)}$ & ${\small 0.41}\overline{{\small 6}}{\small (27.}\overline{{\small %
7}}{\small )}$ & ${\small 0.43}\overline{{\small 3}}{\small (27.1)}$ & $%
{\small 0.42}\overline{{\small 7}}{\small (25.}\overline{{\small 6}}{\small )%
}$ & ${\small 0.414(24.2)}$ \\ 
$\overline{{\small c}}_{3n}$ &  &  & ${\small 0.16}\overline{{\small 6}}%
{\small (11.}\overline{{\small 1}}{\small )}$ & ${\small 0.23}\overline{%
{\small 3}}{\small (14.6)}$ & ${\small 0.26}\overline{{\small 1}}{\small (15.%
}\overline{{\small 6}}{\small )}$ & ${\small 0.271(15.8)}$ \\ 
$\overline{{\small c}}_{4n}$ &  &  &  & ${\small 0.100(06.3)}$ & ${\small %
0.150(9.0)}$ & ${\small 0.176(10.3)}$ \\ 
$\overline{{\small c}}_{5n}$ &  &  &  &  & ${\small 0.06}\overline{{\small 6}%
}{\small (4.0)}$ & ${\small 0.105(06.}\overline{{\small 1}}{\small )}$ \\ 
$\overline{{\small c}}_{6n}$ &  &  &  &  &  & ${\small 0.048(02.}\overline{%
{\small 7}}{\small )}$ \\ 
&  &  &  &  &  &  \\ 
$\overline{v}_{n}$ & ${\small 1.000}$ & ${\small 1.33}\overline{{\small 3}}$
& ${\small 1.500}$ & ${\small 1.600}$ & ${\small 1.66}\overline{{\small 6}}$
& ${\small 1.714}$ \\ \hline
\end{tabular}%
\caption{\footnotesize Ordered authors - credits awarded to each author (the values inside
the brackets are the percentages over the total) and the expected value of the
$n$-authors publication (unit of measure $\overline{v}_{1}=\overline{c}_{1}=1$).}%
\label{tab_ordered}%
\end{table}%

Table \ref{tab_ordered} shows the credits awarded to each author ($\overline{%
c}_{in}$) in an ordered $n$-authors publication for $n=1,...,6,$ according
to the method proposed by \cite{stallings2013determining}. As in this paper,
the \cite{stallings2013determining} approach also follows an axiomatic
approach. Their approach is very practical and can deal with a large variety
of particular cases in terms of authors' ordering. \cite{osorioimpossibility}%
\ presents a detailed study on the properties and limitations of this and
other counting methods.

\section{\textbf{Comparison of theoretical (expected) and empirical values 
\label{sec_compare}}}

In this section, we study whether the expected values in Table \ref%
{tab_results1} make sense. In particular, how do the theoretical values
correlate with empirical values? In this context, we study whether the
theoretical values are in agreement with empirical values based on a large
dataset with bibliometric data. In order to achieve this objective, we use
citations as a measure of \textquotedblleft value\textquotedblright ,
because citations are usually applied to assess the usefulness and the value
of publications for other researchers (\citealp{bornmann2017measuring}).

\subsection{\textbf{Data}}

The bibliometric data used in this paper are from an in-house database
developed and maintained by the Max Planck Digital Library (MPDL, Munich)
and derived from the Science Citation Index Expanded (SCI-E), the Social
Sciences Citation Index (SSCI), and the Arts and Humanities Citation Index
(AHCI) prepared by Clarivate Analytics (see https://clarivate.com), formerly
the IP \& Science business of Thomson Reuters.

The in-house database includes the number of authors for each paper since
1980. In this study, we considered only papers with the document type
\textquotedblleft article\textquotedblright\ to avoid distortion of the
results by the use of papers with different document types. We included all
articles published between 2000 and 2014; more recent years have been
excluded because the window for the citation metrics becomes too small.
Citations are counted in the in-house database from publication year until
the end of 2016.

Two citation-based indicators are considered in this study to empirically
assess the value of the $n$-authors publications. Citations are one of the
most frequently used metrics in research evaluations which reflect the
impact of publications (as one part of quality) (%
\citealp{martin1983assessing}). The empirical analysis is based on articles
from different fields. Since researchers in different fields have different
publication and citation cultures, it is standard in bibliometrics to apply
field-normalized citations scores for the impact comparison of papers from
different fields. In the following, we use two different field-normalized
indicators which are standard in bibliometrics to cross-check the results (%
\citealp{bornmann2015methods}):

For the normalized citation score (NCS), the citation counts of an article
are divided by the expected citation impact. For the calculation of the
expected impact, the average citation rate is calculated based on all
articles which have been published in the same year and field as the focal
article. To aggregate the citation impact of more than one paper (e.g., all
papers with one author), the arithmetic average was used. Percentiles offer
an alternative to the mean-based quotient NCS. A percentile is a value below
which a certain proportion of publications fall. For the calculation of
percentiles, all papers in a field and publication year were ranked in
decreasing order by their number of citations. Then, the percentiles have
been calculated according to the formula $(i-0.5)/n\times 100$---whereby $i$
is the rank number and $n$ the number of papers in the set (%
\citealp{hazen1914storage}). The percentiles for a set of papers (e.g., all
papers with one author) have been aggregated by using the median.

As field classification scheme for normalizing impact, the WoS subject
categories have been used for both indicators in this study. These subject
categories are based on sets of journals that publish papers in similar
research areas.\qquad

\begin{figure}[h]
\begin{center}
\includegraphics[width=.7\textwidth]{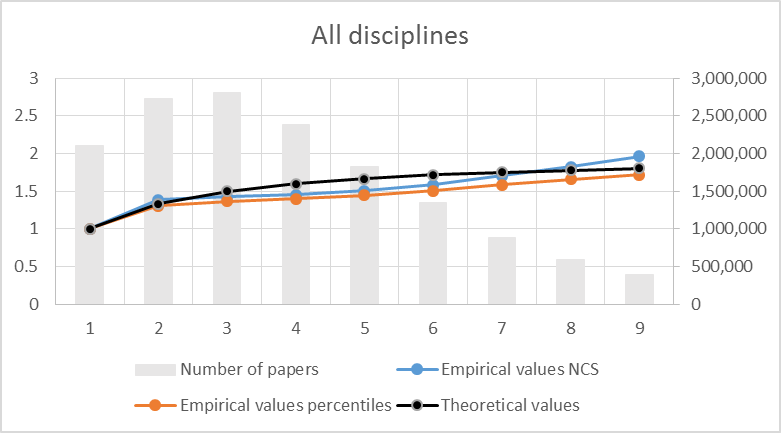}
\end{center}
\caption{{\protect\footnotesize Comparison of the theoretically derived
expected values for different numbers of authors with the empirically
derived field-normalized citation scores (NCS and percentiles). The figure
also shows the number of papers with different numbers of authors.}}
\label{figure_main}
\end{figure}

\subsection{\textbf{Results}}

Figure \ref{figure_main} shows the comparison of the theoretically derived
publication values for different numbers of authors with the empirically
derived field-normalized citation scores (NCS and percentiles). The figure
does not visualize the field-normalized citation scores, but relative
values, which show how the citation impact varies with different numbers of
authors. For example, the visualized NCS value ($1.38$) for two authors has
been calculated by multiplying the (empirical) NCS=0.895 (for two authors)
with $v_{1}=1$ (see Table \ref{tab_results1}) and dividing the product by
the empirical NCS score for one author (NCS=0.647).

As the results in Figure \ref{figure_main} over all disciplines show, the
empirical and theoretical values are in close agreement. Since the WoS
database is mostly based on natural sciences publications, the results in
Figure \ref{figure_main} are especially driven by these publications. Other
disciplines have not only different publication and citation cultures, but
also different treatments of authorship positions.

Thus, we produced further results for six broad disciplines: (1) natural
sciences, (2) engineering and technology, (3) medical and health sciences,
(4) agricultural sciences, (5) social sciences, and (6) humanities. The
broad disciplines are aggregated paper sets of WoS subject categories based
on the OECD category scheme (major codes). The results for the different
disciplines are shown in Figure \ref{figure_all}. As expected, the results
for the natural sciences are similar to the results in Figure \ref%
{figure_main}. Since Figure \ref{figure_all} also includes the number of
papers with the different numbers of authors, it is clearly visible that the
WoS database is mainly based on papers from the natural sciences.

\begin{figure}[tbp]
\begin{subfigure}[b]{0.5\textwidth}
    \includegraphics[width=\textwidth]{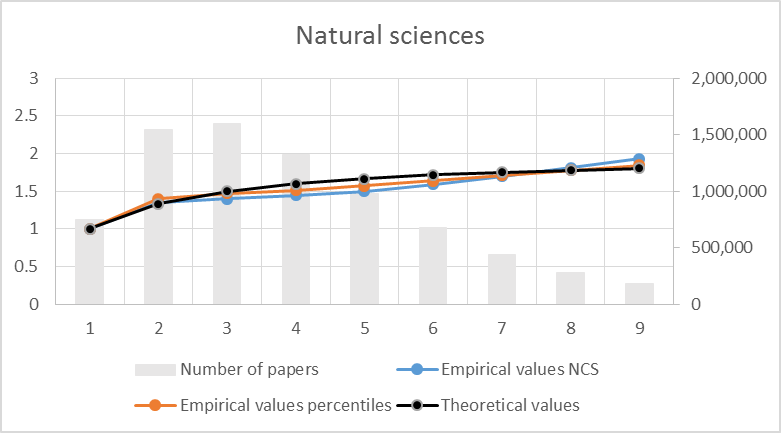}
    \label{figRH}
  \end{subfigure}\hfill 
\begin{subfigure}[b]{0.5\textwidth}
    \includegraphics[width=\textwidth]{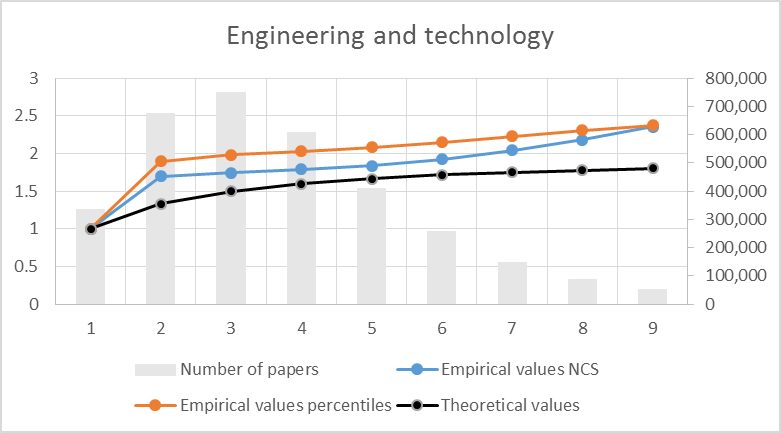}
    \label{figPH}
  \end{subfigure}\linebreak 
\begin{subfigure}[b]{0.5\textwidth}
    \includegraphics[width=\textwidth]{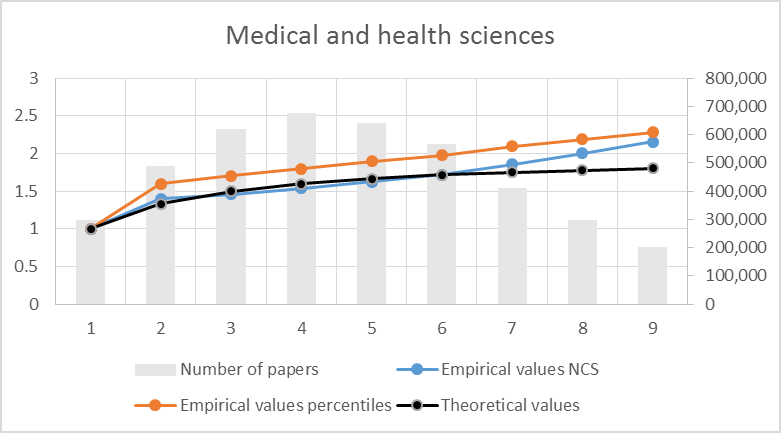}
    \label{figIH}
  \end{subfigure}\hfill 
\begin{subfigure}[b]{0.5\textwidth}
    \includegraphics[width=\textwidth]{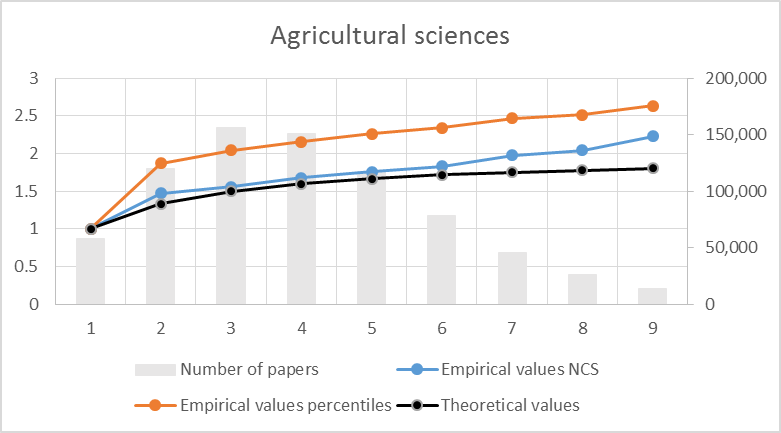}
    \label{figQH}
  \end{subfigure}\linebreak 
\begin{subfigure}[b]{0.5\textwidth}
    \includegraphics[width=\textwidth]{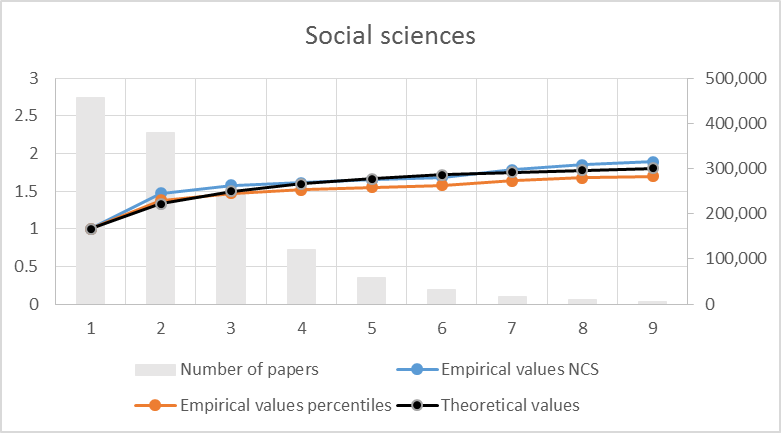}
    \label{figRH}
  \end{subfigure}\hfill 
\begin{subfigure}[b]{0.5\textwidth}
    \includegraphics[width=\textwidth]{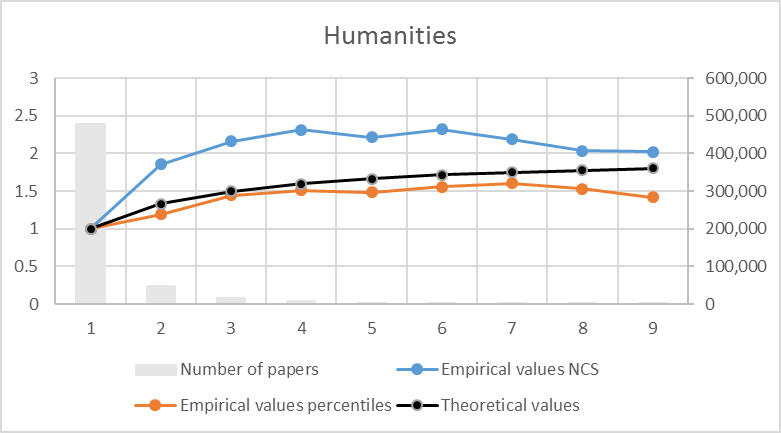}
    \label{figPH}
  \end{subfigure}\linebreak
\caption{{\protect\footnotesize Comparison of the theoretically derived
publication values for different numbers of authors with the empirically
derived field-normalized citation scores (NCS and percentiles) for six broad
disciplines. The figure also shows the number of papers with different
numbers of authors.}}
\label{figure_all}
\end{figure}

The number of papers for the different numbers of authors reveals that the
authorship cultures are different among the disciplines. In natural
sciences, engineering and technology, medical and health sciences, as well
as agricultural sciences, we see peaks at around two or three authors. In
the social sciences and humanities, the publication with only one author is
the most frequent publication type (especially in the humanities).

The comparison of the empirical with the theoretical values in Figure \ref%
{figure_all} indicates that the highest agreement is obtained in natural and
social sciences (although the social sciences show a different authorship
pattern than the natural sciences). In medical and health sciences, the NCS
values are close to the theoretical values, but the percentiles differ. For
engineering and technology as well as agricultural sciences, the percentile
values in particular differ from the theoretical values. The greatest
difference between empirical evidence and theory can be seen for the
humanities. We also produced the results for the OECD minor codes to present
more detailed field-specific results. The results can be found in \ref%
{sec_appendix_figures}. The more detailed results confirm the results from
the higher aggregation level (the major codes).

\bigskip

To sum up, the empirical results indicate that the theoretical values are
close to the empirical values---if citation impact is used to assess the
value of publications. However, there are differences between the
disciplines: while the natural and social sciences are close to the
expectations, the humanities show some differences. The other disciplines
demonstrate reasonably good agreement between theoretical and empirical
results.

\section{\textbf{Conclusion\label{sec_concl}}}

Multidisciplinary scientific collaboration is increasing (%
\citealp{gazni2012mapping}; \citealp{katz1997research}; %
\citealp{lariviere2015team}; \citealp{persson2004inflationary}; %
\citealp{wuchty2007increasing}). However, it is extremely difficult for
third parties (e.g., a panel or an evaluation committee) to quantify the
contribution of each author. In this context, we should be able to find ways
to account and distinguish between publications with different number of
authors in order to obtain evaluations that are more suitable.

This paper attempts to quantify the expected value of an $n$-authors
publication. Publications with several authors tend to have more impact
(e.g., in terms of citations) than single-author publications, maybe because
of the possibility of synergies and cross-fertilization of ideas (%
\citealp{hsu2011correlation}; \citealp{onodera2015factors}; among others).
However, the strength of this effect is limited because of coordination
difficulties, free-riding, and other forms of opportunistic behavior. A
greater number of authors may lead to larger aggregate efforts, but not
necessarily to larger individual efforts, which tends to decrease with the
number of authors.

In this paper, we propose a set of principles and axioms that we consider
fundamental to determine the expected value of an $n$-authors publication.
The result is a unique measure of the expected value and credits of $n$%
-authors publications. The expected value of the publication increases
monotonically with the number of authors, but at a decreasing rate because
the marginal contribution of new co-authors becomes less and less
significant.

Using a comprehensive set of bibliometric data, we found that the
theoretically obtained expected values and patterns are close to the
empirical values for some disciplines. These results provide support in
favor of the method proposed in this paper. However, these results also make
explicit that this method should not be taken as a universal solution, which
can be applied indiscriminately to quantify the expected value of $n$%
-authors publications in all disciplines.

\bigskip

The proposed approach follows a set of principles or axioms that we consider
fundamental to derive the expected value of $n$-authors publications.
However,\ this approach or any other approach obtained according to any
other principles or axioms will always be a subject of discussion. As a rule
it is difficult to agree on one particular method (%
\citealp{waltman2016review}). There are several reasons for this lack of
agreement. First, the expected value of the publication and the associated
counting method play a crucial role in academics' and scientists' lives.
Second, judgements and evaluations of qualitative issues like those
concerning scientific publications are always subjective and open to debate,
which creates consensus difficulties and allows the coexistence of different
approaches. However, our proposed solution is practical and has the
advantage of not requiring information other than the number of authors.

\bigskip\ 

Finally, this study suggests a relation between the expected value of the
publication and the number of authors; an aspect that has been discussed,
but not formally addressed in the literature (\citealp{hsu2011correlation}; %
\citealp{onodera2015factors}; among others). In this context, we expect that
our findings can help researchers and decision-makers to choose and
implement more effective and fair counting methods that take into account
the benefits of collaboration.

\bigskip

\bigskip

\bigskip

\noindent {\footnotesize {\textbf{Acknowledgments:}} We wish to thank to
Ricardo Ribeiro, Juan Pablo Rinc\'{o}n-Zapatero and Ludo Waltman for helpful
comments and discussions. Financial support from the Spanish Ministerio of
Ciencia y Innovaci\'{o}n project ECO2016-75410-P, GRODE and the Barcelona
GSE is gratefully acknowledged.}

{\footnotesize The bibliometric data used in this paper is from an in-house
database developed and maintained by the Max Planck Digital Library (MPDL,
Munich) and derived from the Science Citation Index Expanded (SCI-E), Social
Sciences Citation Index (SSCI), Arts and Humanities Citation Index (AHCI)
prepared by Clarivate Analytics (see https://clarivate.com), formerly the IP
\& Science business of Thomson Reuters (Philadelphia, Pennsylvania, USA).}

\bigskip

\pagebreak \appendix

\section{\textbf{Proofs of the Propositions}\label{sec_appendix_proofs}}

\begin{proof}[Proof of Proposition \protect\ref{prop_equal_credits} and
Corollary \protect\ref{coro_satisfies}]
In order to show Proposition \ref{prop_equal_credits}, we aggregate the
mathematical implications of Axioms \ref{colab_funct}-\ref{incent_compat}.
This construction will lead to expression (\ref{exp_equal_credits}).\ Axioms %
\ref{colab_funct} and \ref{unif_dist} establish that the value of a $n$%
-authors collaboration is the sum of the $n$ authors' expected efforts,
i.e., $\overline{v}_{n}=\sum\nolimits_{i=1}^{n}E(e_{in})$ for $%
n=1,...,\infty .$\ In addition, the independent and uniform assumption of
Axiom \ref{unif_dist} implies the continuous probability density function
(PDF) $f(e_{in})=1/\overline{c}_{n-1}$ for $i=1,...,n$ and $n=1,...,\infty ,$
which in the case of $n$ independent random variables implies the joint
probability density function $\prod\nolimits_{i=1}^{n}f(e_{in})$ for $%
n=1,...,\infty .$ Axiom \ref{aggregated_eff} requires that the aggregated
effort of an $n$-authors collaboration must be above some threshold, i.e., $%
\sum\nolimits_{i=1}^{n}e_{in}\geq \overline{v}_{n-1}$ for $n=2,...,\infty .$%
\ Finally, Axiom \ref{incent_compat} establishes that in order for an $n$%
-authors collaboration to be stable and each author to have incentives to
participate in it, the individual effort must also satisfy $e_{in}\in
\lbrack 0,\overline{e}_{n-1}]$ for $i=1,...,n$ and $n=2,...,\infty ,$ where $%
\overline{e}_{n-1}=\overline{c}_{n-1}=\overline{v}_{n-1}/(n-1).$\ The
consideration of Axioms \ref{colab_funct}-\ref{incent_compat} implies that
the value of an $n$-authors collaboration is by construction uniquely given
by the following conditional expectation:%
\begin{eqnarray}
\overline{v}_{n} &=&E(\sum\nolimits_{i=1}^{n}e_{in}|H)=E(\mathbf{1}%
_{H}\sum\nolimits_{i=1}^{n}e_{in})/P(H)  \notag \\
&=&\int\nolimits_{(e_{1n},e_{2n},...,e_{nn})\in \lbrack 0,\overline{v}%
_{n-1}/(n-1)]^{n}}\sum%
\nolimits_{i=1}^{n}e_{in}dP(e_{1n},e_{2n},...,e_{nn}|H),
\label{exp_general_integ}
\end{eqnarray}%
for $n=2,...,\infty ,$ where the expectation is taken with respect to the
joint density function of the $n$-dimensional vector of efforts $%
(e_{1n},e_{2n},...,e_{nn})\in \lbrack 0,\overline{c}_{n-1}]^{n},$ and
conditional on the event $H=\{\sum\nolimits_{i=1}^{n}e_{in}\geq \overline{v}%
_{n-1}\},$\ where $P(H)$ denotes the probability of the event $H,$ and $%
\mathbf{1}_{H}$\ is an indicator function that takes the value $1$ when the
event $H$ occurs, and $0$\ otherwise. Consequently, the integral in
expression (\ref{exp_general_integ}) can be rewritten in the following
equivalent way: 
\begin{equation}
\overline{v}_{n}=\frac{\int\nolimits_{0}^{\overline{c}_{n-1}}f(e_{1n})...%
\int\nolimits_{0}^{\overline{c}_{n-1}}f(e_{nn})\mathbf{1}_{H}(\sum%
\nolimits_{i=1}^{n}e_{in})de_{nn}...de_{1n}}{\int\nolimits_{0}^{\overline{c}%
_{n-1}}f(e_{1n})...\int\nolimits_{0}^{\overline{c}_{n-1}}f(e_{nn})\mathbf{1}%
_{H}de_{nn}...de_{1n}},  \label{exp_integral1}
\end{equation}%
for $n=2,...,\infty .$ Note also that since $f(e_{in})=1/\overline{c}_{n-1}$
is constant and independent of $e_{in},$ we can trivially cancel the
numerator by the denominator. However, in order to solve this integral
analytically, the indicator function $\mathbf{1}_{H}$ in expression (\ref%
{exp_integral1}) must be passed to the integration limits. In this context,
we must rewrite expression (\ref{exp_integral1}) in the following equivalent
way, in which author $1$ provides more effort than author $2$, and so on in
decreasing order until author $n,$ i.e.:%
\begin{equation}
\overline{v}_{n}=\frac{n!\int\nolimits_{\frac{\overline{v}_{n-1}}{n}}^{%
\overline{c}_{n-1}}\int\nolimits_{\frac{\overline{v}_{n-1}-e_{1n}}{n-1}%
}^{e_{1n}}...\int\nolimits_{\frac{\overline{v}_{n-1}-e_{1n}...-e_{n-1n}}{1}%
}^{e_{n-1n}}(\sum\nolimits_{i=1}^{n}e_{in})de_{nn}...de_{2n}de_{1n}}{%
n!\int\nolimits_{\frac{\overline{v}_{n-1}}{n}}^{\overline{c}%
_{n-1}}\int\nolimits_{\frac{\overline{v}_{n-1}-e_{1n}}{n-1}%
}^{e_{1n}}...\int\nolimits_{\frac{\overline{v}_{n-1}-e_{1n}...-e_{n-1n}}{1}%
}^{e_{n-1n}}de_{nn}...de_{2n}de_{1n}},  \label{exp_integral2}
\end{equation}%
for $n=2,...,\infty ,$ where $n!$ is multiplying the numerator and the
denominator to denote that each of the $n$ authors must be in each of the $%
n! $ possible effort ordered permutations in order for expression (\ref%
{exp_integral2}) to be equivalent to expression (\ref{exp_integral1}). The
integration lower bound guarantees that condition $H$ is satisfied at all
steps of integration. For instance, for $e_{nn}$ we must have $e_{nn}\geq 
\frac{\overline{v}_{n-1}-e_{1n}-e_{2n}...-e_{n-1n}}{1},$ while for $e_{n-1n}$
we must have $e_{n-1n}\geq \frac{\overline{v}_{n-1}-e_{1n}-e_{2n}...-e_{n-2n}%
}{2},$ and so on. For instance, given the efforts provided by author $1$
until author $n-2$ (i.e., $e_{1n},$ $e_{2n},$ $...,e_{n-2n}$), and since
author $n-1$ cannot provide less effort than the last author $n$, the
minimum effort of author $n-1$ occurs when author $n-1$ and author $n$
divide equally the effort that is still left in order to satisfy condition $%
H,$ i.e., at $e_{n-1n}=\frac{\overline{v}_{n-1}-e_{1n}-e_{2n}...-e_{n-2n}}{2}%
.$ After this transformation, the integrals in the numerator and denominator
of expression (\ref{exp_integral2}) can be solved analytically. After some
algebra, we obtain that:%
\begin{equation}
\overline{v}_{n}=\frac{n^{2}}{(n-1)(n+1)}\overline{v}_{n-1},
\label{exp_recur1}
\end{equation}%
for $n=2,...,\infty ,$ where we have made use of the fact that $\overline{c}%
_{n-1}=\overline{v}_{n-1}/(n-1).$ Therefore, if we know the value of $%
\overline{v}_{1},$ we can recursively obtain the expression (\ref%
{exp_equal_credits}) of Proposition \ref{prop_equal_credits}. Thus, by
simply dividing $\overline{v}_{n}$ by the number of authors, we obtain that
the credits $\overline{c}_{n}$ awarded to each author when all authors are
equally important are given by:%
\begin{equation*}
\overline{c}_{n}=\frac{2(n!)^{2}}{n(n-1)!(n+1)!}\overline{v}_{1}=\frac{2}{n+1%
}\overline{v}_{1},
\end{equation*}%
for $n=1,...,\infty ,$ where \textquotedblleft $!$" denotes the factorial
symbol.

In order to show Corollary \ref{coro_satisfies},\ it is enough to show that
expression (\ref{exp_recur1}) falls inside the bounds defined by inequality (%
\ref{basic_ineq}), which is implied by inequalities (\ref{chain_value}) and (%
\ref{chain_average}). The lower bound is satisfied if $\overline{v}_{n}\geq 
\overline{v}_{n-1},$ i.e., if $n^{2}\geq n^{2}-1$ which is always true. The
upper bound is satisfied if $\overline{v}_{n}\leq n\overline{v}_{n-1}/(n-1),$
i.e., if $n^{2}(n-1)\leq n(n^{2}-1)$ which is also always true.
\end{proof}

\pagebreak

\section{\textbf{Field-specific results (OECD minor codes)}\label%
{sec_appendix_figures}}

\begin{figure}[!b]
\begin{subfigure}[b]{0.33333\textwidth}
    \includegraphics[width=\textwidth]{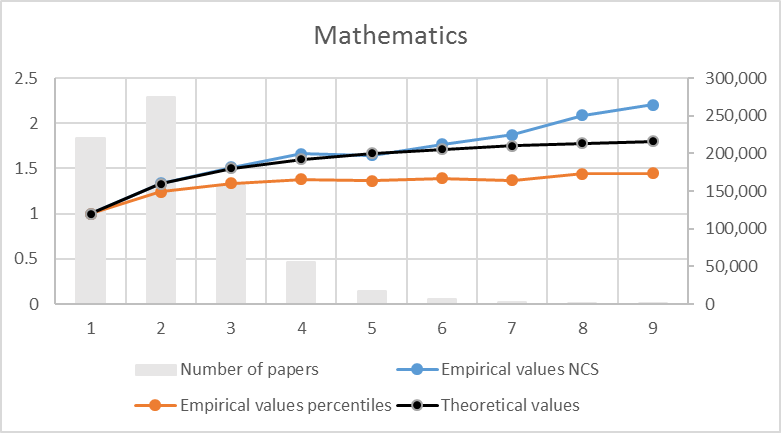}
    \label{figRH}
  \end{subfigure}%
\begin{subfigure}[b]{0.33333\textwidth}
    \includegraphics[width=\textwidth]{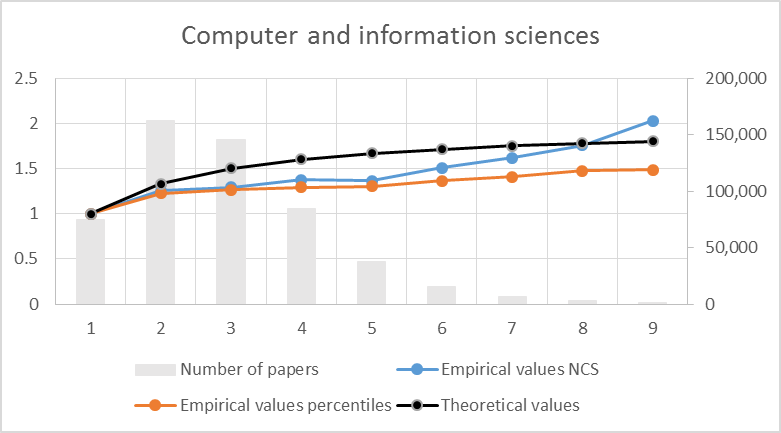}
    \label{figRH}
  \end{subfigure}\hfill 
\begin{subfigure}[b]{0.33333\textwidth}
    \includegraphics[width=\textwidth]{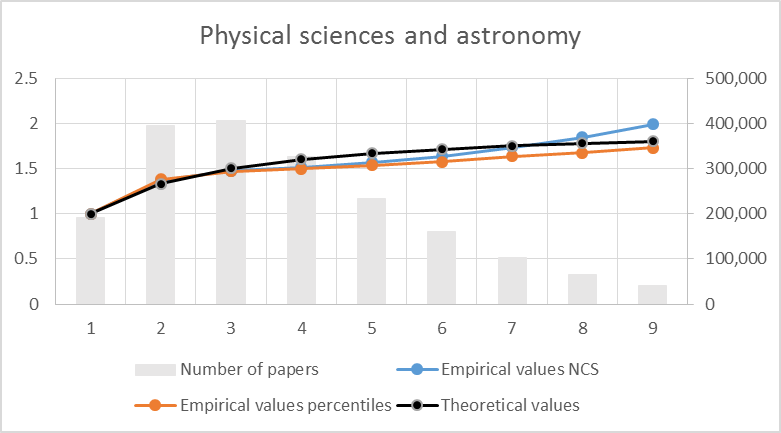}
    \label{figPH}
  \end{subfigure}\linebreak 
\begin{subfigure}[b]{0.33333\textwidth}
    \includegraphics[width=\textwidth]{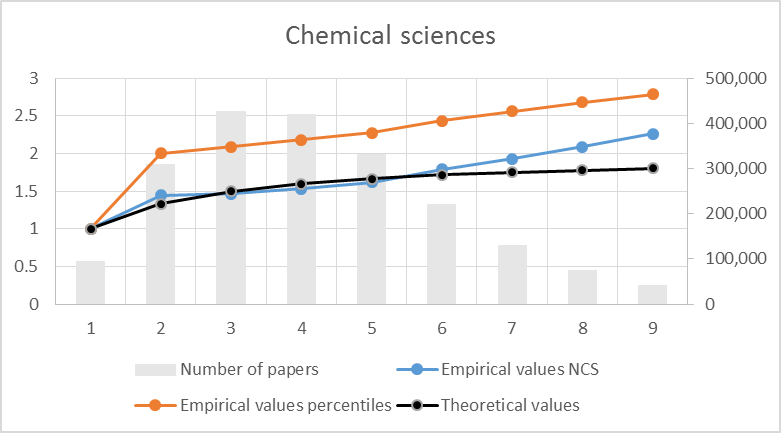}
    \label{figRH}
  \end{subfigure}%
\begin{subfigure}[b]{0.33333\textwidth}
    \includegraphics[width=\textwidth]{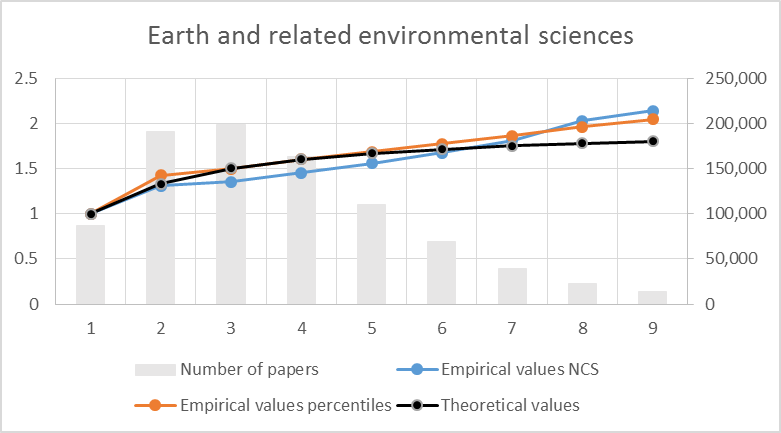}
    \label{figRH}
  \end{subfigure}\hfill 
\begin{subfigure}[b]{0.33333\textwidth}
    \includegraphics[width=\textwidth]{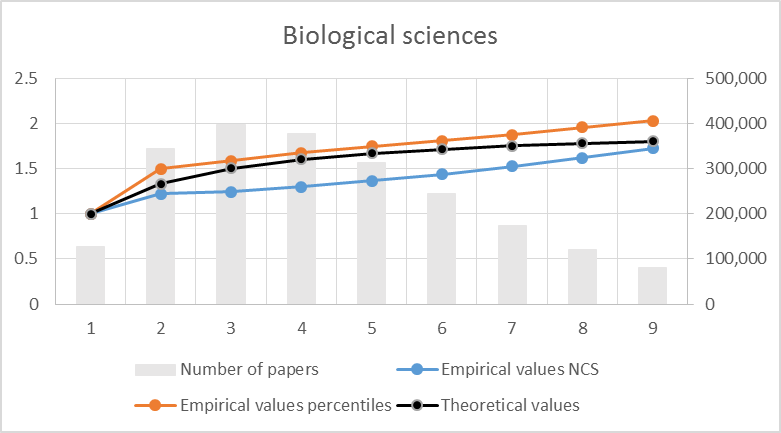}
    \label{figPH}
  \end{subfigure}\linebreak 
\begin{subfigure}[b]{0.33333\textwidth}
    \includegraphics[width=\textwidth]{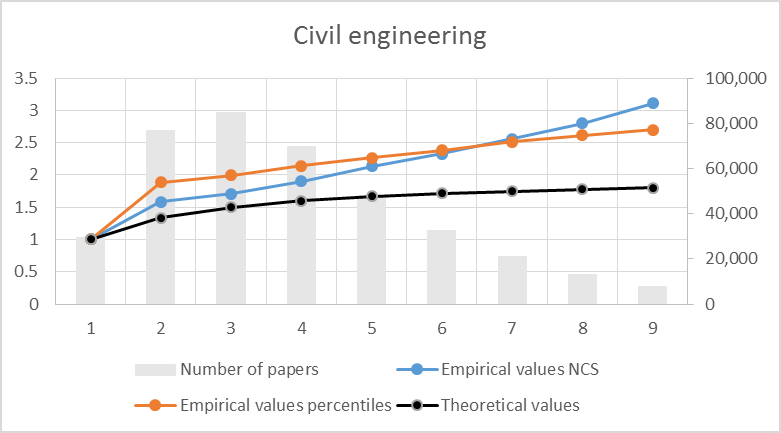}
    \label{figRH}
  \end{subfigure}%
\begin{subfigure}[b]{0.33333\textwidth}
    \includegraphics[width=\textwidth]{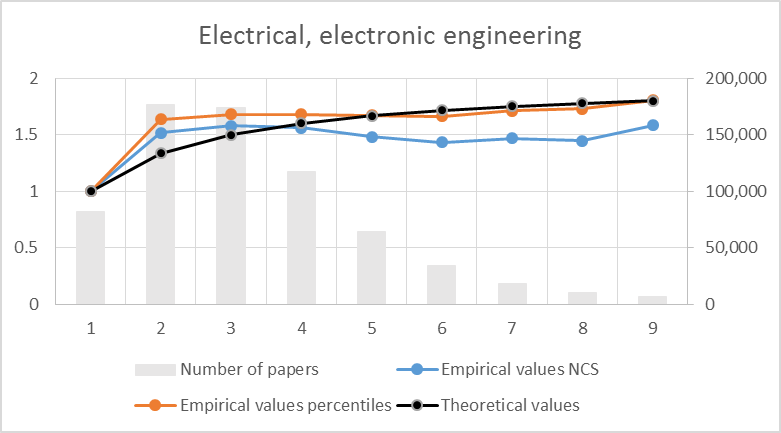}
    \label{figRH}
  \end{subfigure}\hfill 
\begin{subfigure}[b]{0.33333\textwidth}
    \includegraphics[width=\textwidth]{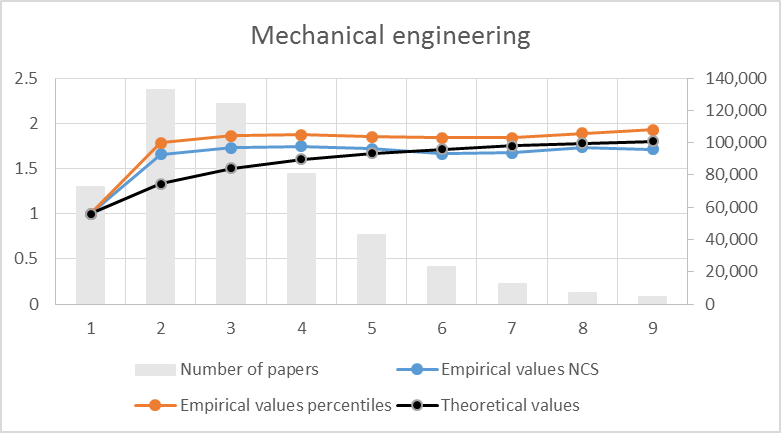}
    \label{figPH}
  \end{subfigure}\linebreak 
\begin{subfigure}[b]{0.33333\textwidth}
    \includegraphics[width=\textwidth]{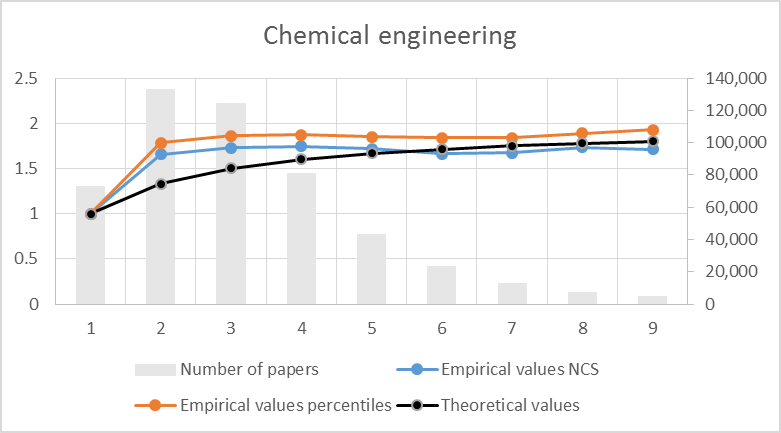}
    \label{figRH}
  \end{subfigure}%
\begin{subfigure}[b]{0.33333\textwidth}
    \includegraphics[width=\textwidth]{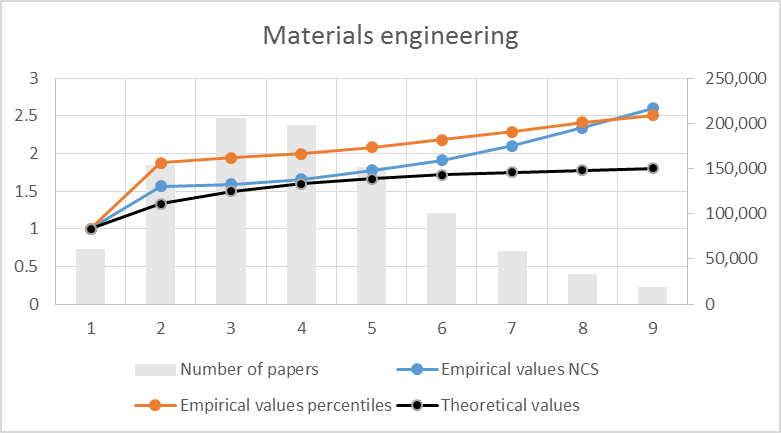}
    \label{figRH}
  \end{subfigure}\hfill 
\begin{subfigure}[b]{0.33333\textwidth}
    \includegraphics[width=\textwidth]{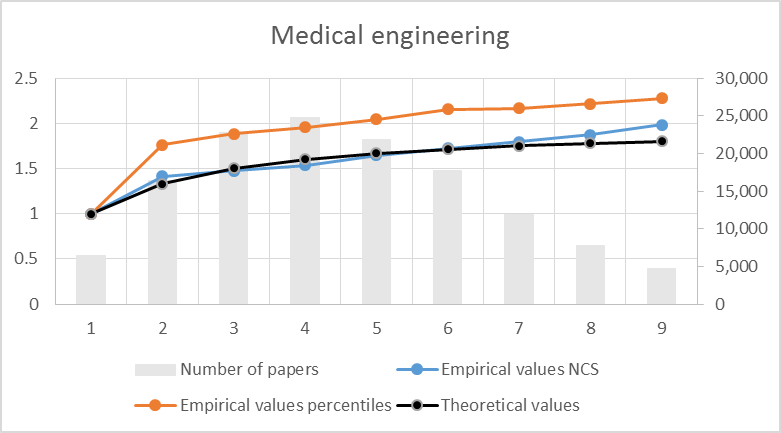}
    \label{figPH}
  \end{subfigure}\linebreak 
\begin{subfigure}[b]{0.33333\textwidth}
    \includegraphics[width=\textwidth]{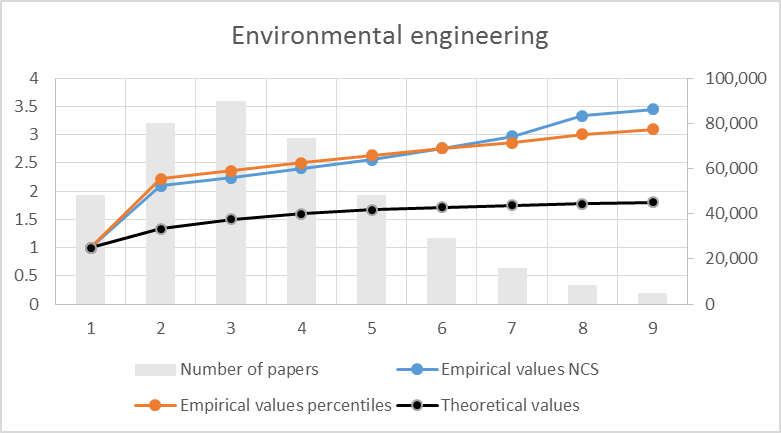}
    \label{figRH}
  \end{subfigure}%
\begin{subfigure}[b]{0.33333\textwidth}
    \includegraphics[width=\textwidth]{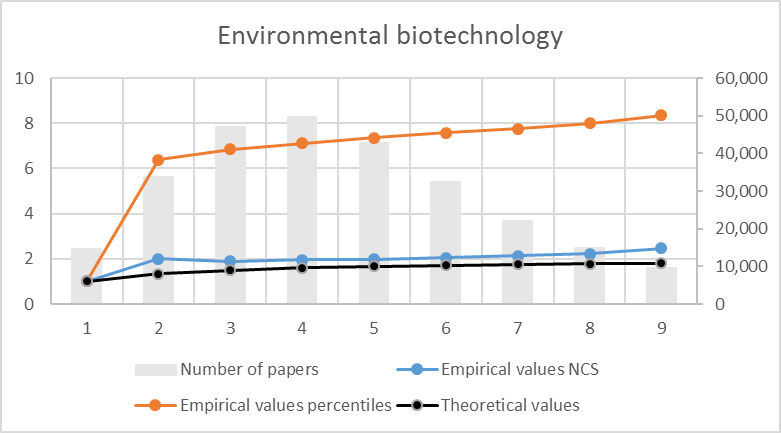}
    \label{figRH}
  \end{subfigure}\hfill 
\begin{subfigure}[b]{0.33333\textwidth}
    \includegraphics[width=\textwidth]{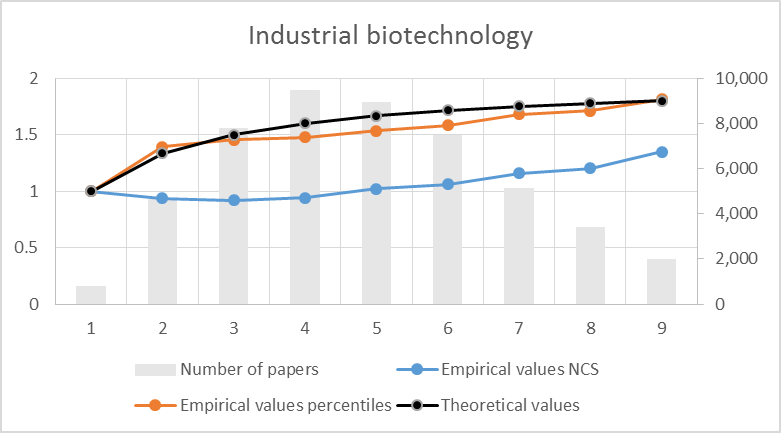}
    \label{figPH}
  \end{subfigure}\linebreak
\end{figure}

\begin{figure}[tbp]
\begin{subfigure}[b]{0.33333\textwidth}
    \includegraphics[width=\textwidth]{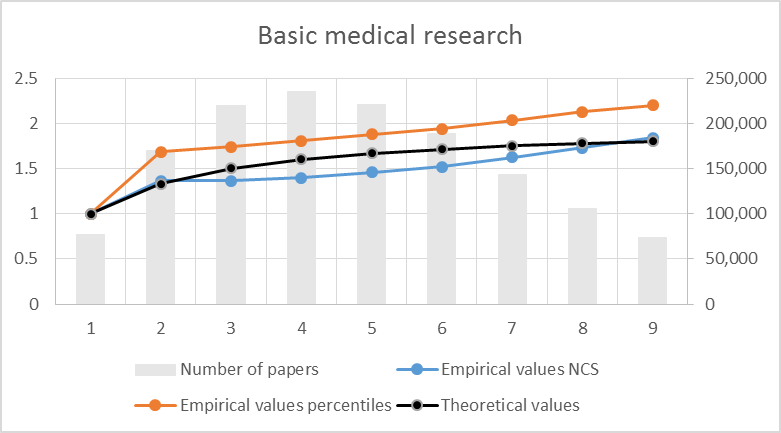}
    \label{figRH}
  \end{subfigure}%
\begin{subfigure}[b]{0.33333\textwidth}
    \includegraphics[width=\textwidth]{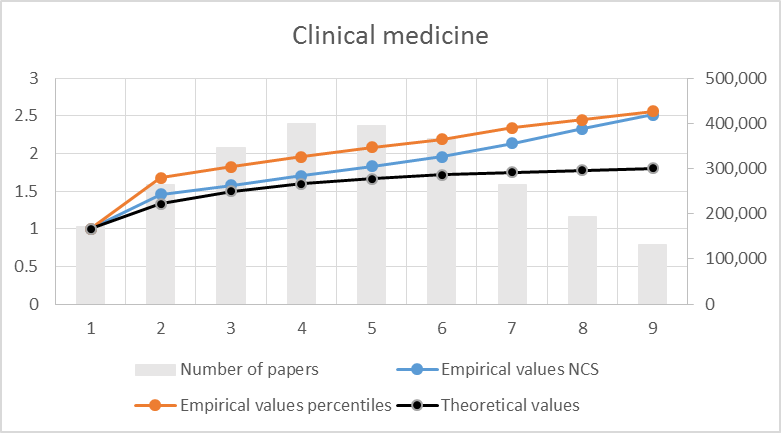}
    \label{figRH}
  \end{subfigure}\hfill 
\begin{subfigure}[b]{0.33333\textwidth}
    \includegraphics[width=\textwidth]{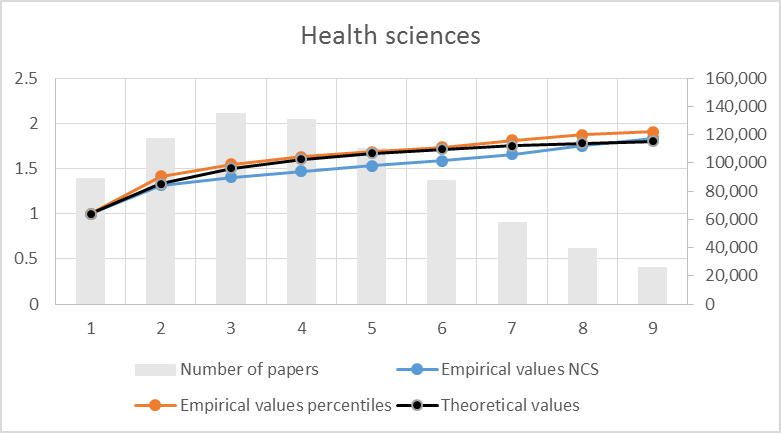}
    \label{figPH}
  \end{subfigure}\linebreak 
\begin{subfigure}[b]{0.33333\textwidth}
    \includegraphics[width=\textwidth]{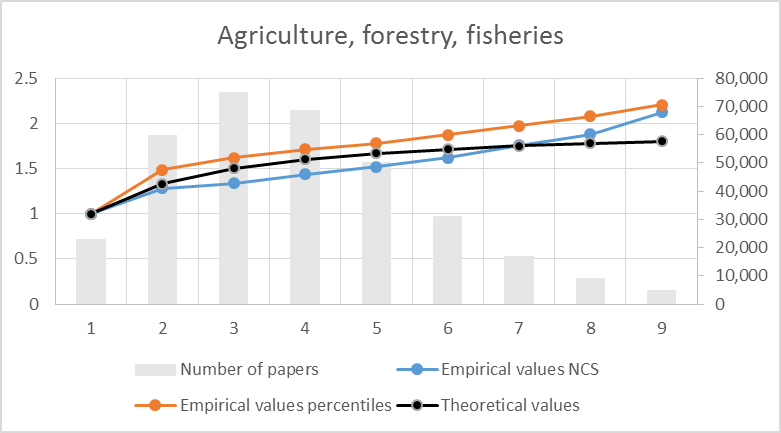}
    \label{figRH}
  \end{subfigure}%
\begin{subfigure}[b]{0.33333\textwidth}
    \includegraphics[width=\textwidth]{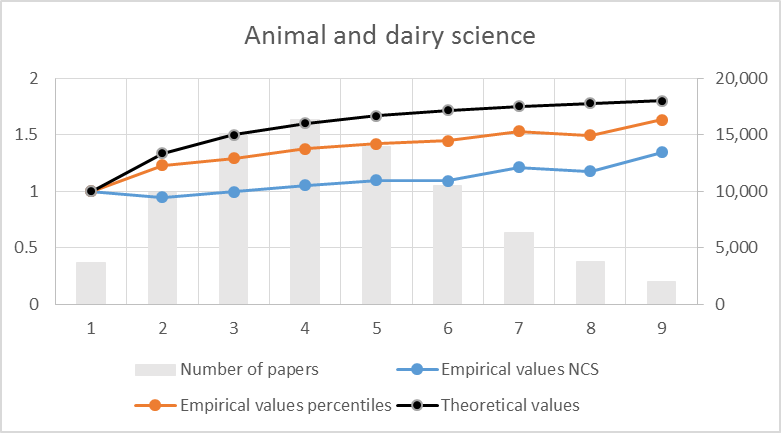}
    \label{figRH}
  \end{subfigure}\hfill 
\begin{subfigure}[b]{0.33333\textwidth}
    \includegraphics[width=\textwidth]{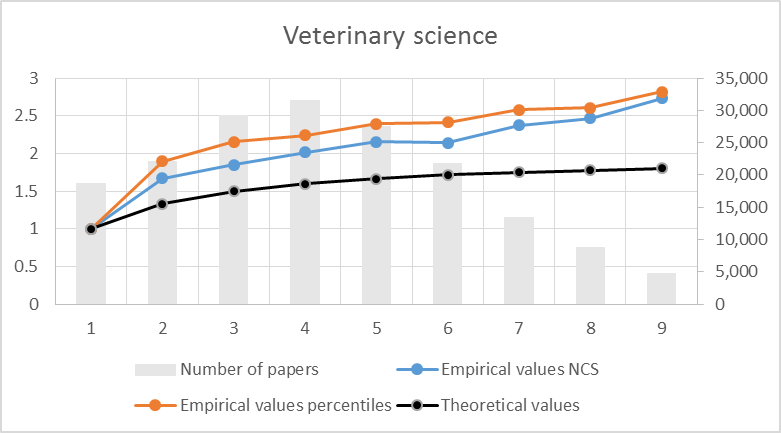}
    \label{figPH}
  \end{subfigure}\linebreak 
\begin{subfigure}[b]{0.33333\textwidth}
    \includegraphics[width=\textwidth]{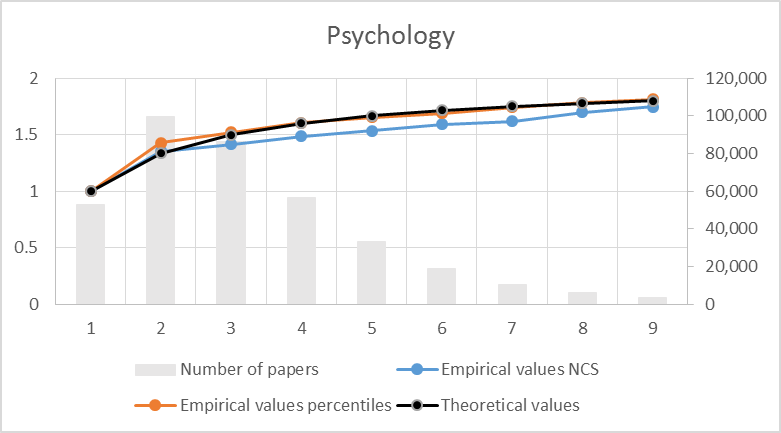}
    \label{figRH}
  \end{subfigure}%
\begin{subfigure}[b]{0.33333\textwidth}
    \includegraphics[width=\textwidth]{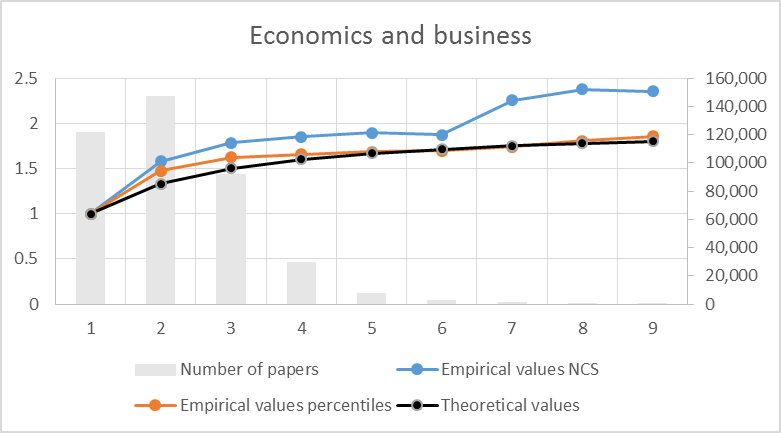}
    \label{figRH}
  \end{subfigure}\hfill 
\begin{subfigure}[b]{0.33333\textwidth}
    \includegraphics[width=\textwidth]{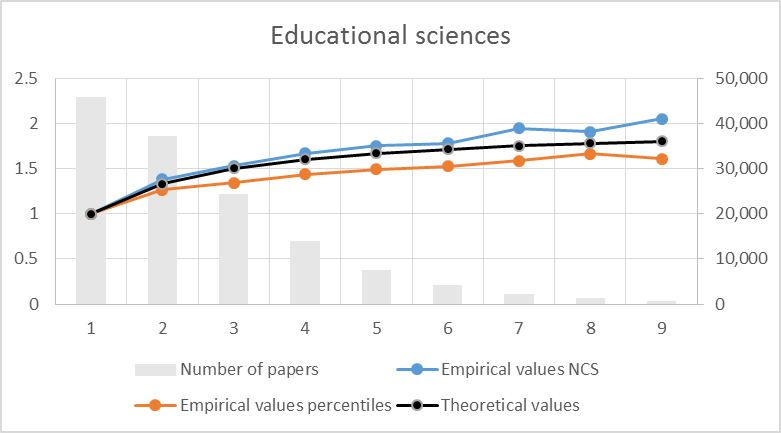}
    \label{figPH}
  \end{subfigure}\linebreak 
\begin{subfigure}[b]{0.33333\textwidth}
    \includegraphics[width=\textwidth]{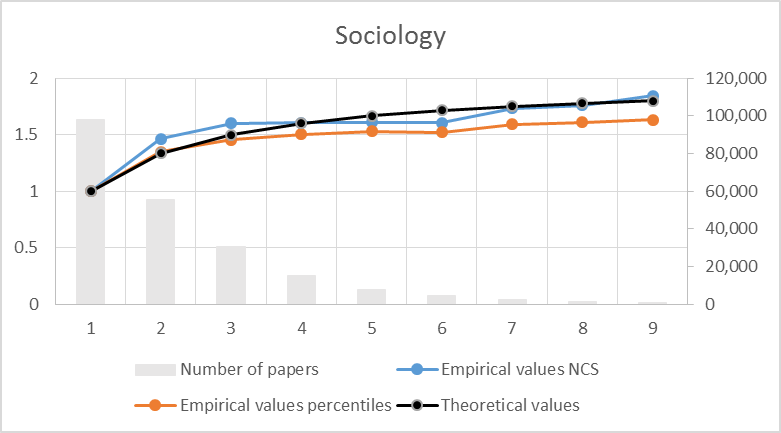}
    \label{figRH}
  \end{subfigure}%
\begin{subfigure}[b]{0.33333\textwidth}
    \includegraphics[width=\textwidth]{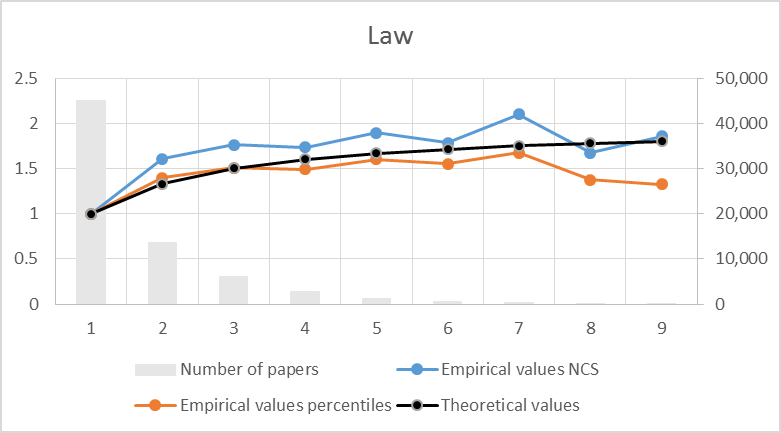}
    \label{figRH}
  \end{subfigure}\hfill 
\begin{subfigure}[b]{0.33333\textwidth}
    \includegraphics[width=\textwidth]{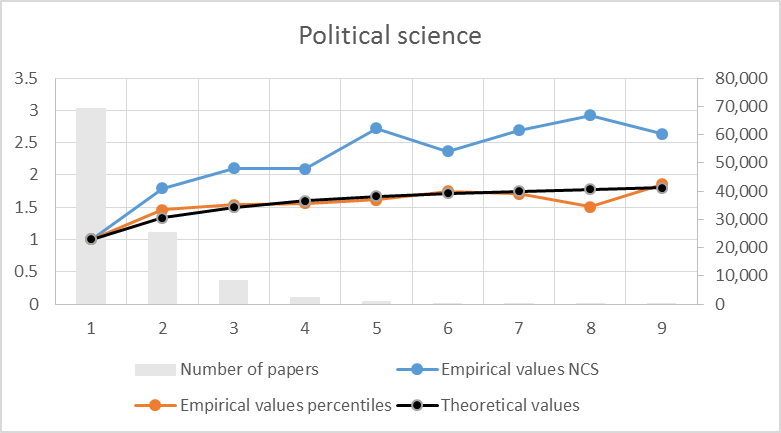}
    \label{figPH}
  \end{subfigure}\linebreak 
\begin{subfigure}[b]{0.33333\textwidth}
    \includegraphics[width=\textwidth]{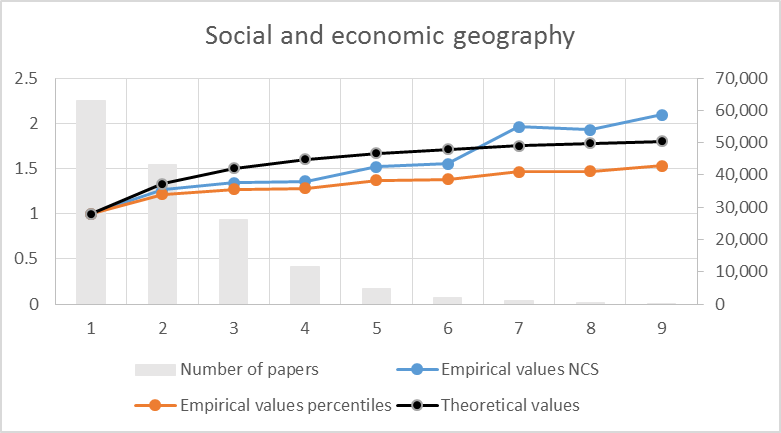}
    \label{figRH}
  \end{subfigure}%
\begin{subfigure}[b]{0.33333\textwidth}
    \includegraphics[width=\textwidth]{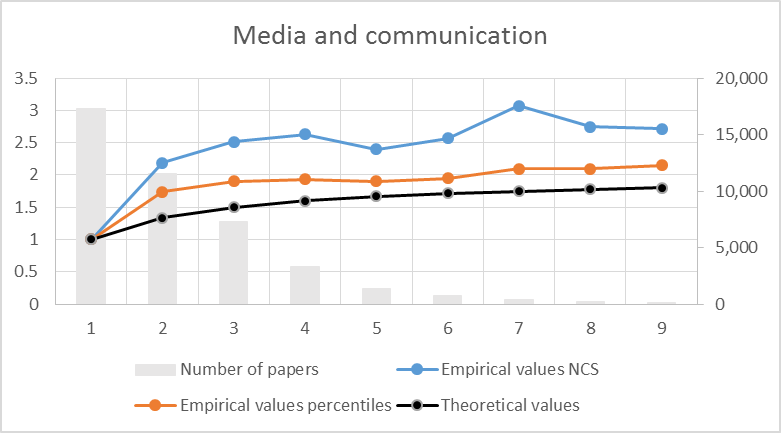}
    \label{figRH}
  \end{subfigure}\hfill 
\begin{subfigure}[b]{0.33333\textwidth}
    \includegraphics[width=\textwidth]{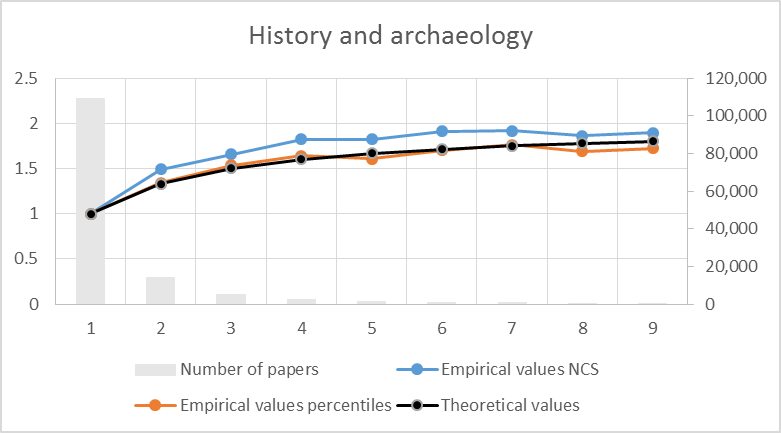}
    \label{figPH}
  \end{subfigure}\linebreak
\end{figure}

\begin{figure}[tbp]
\begin{subfigure}[b]{0.33333\textwidth}
    \includegraphics[width=\textwidth]{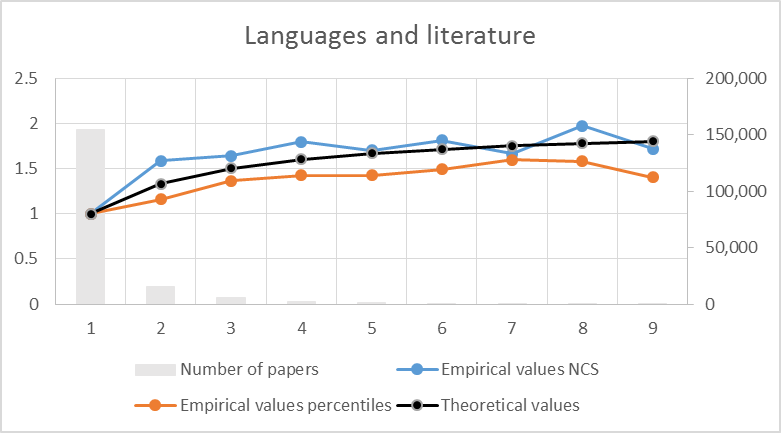}
    \label{figRH}
  \end{subfigure}%
\begin{subfigure}[b]{0.33333\textwidth}
    \includegraphics[width=\textwidth]{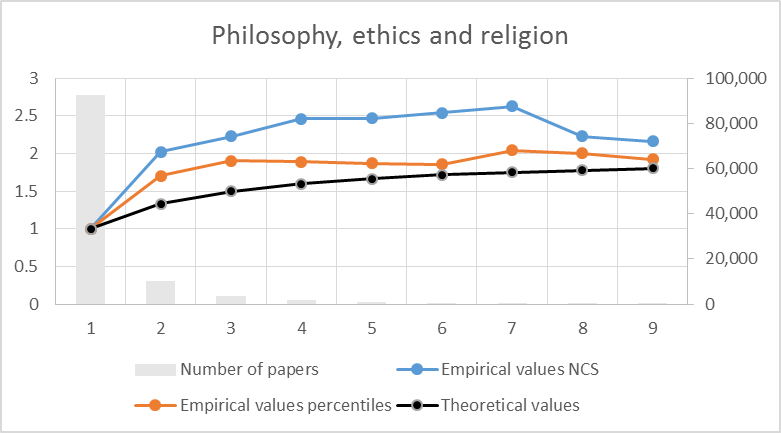}
    \label{figRH}
  \end{subfigure}\hfill 
\begin{subfigure}[b]{0.33333\textwidth}
    \includegraphics[width=\textwidth]{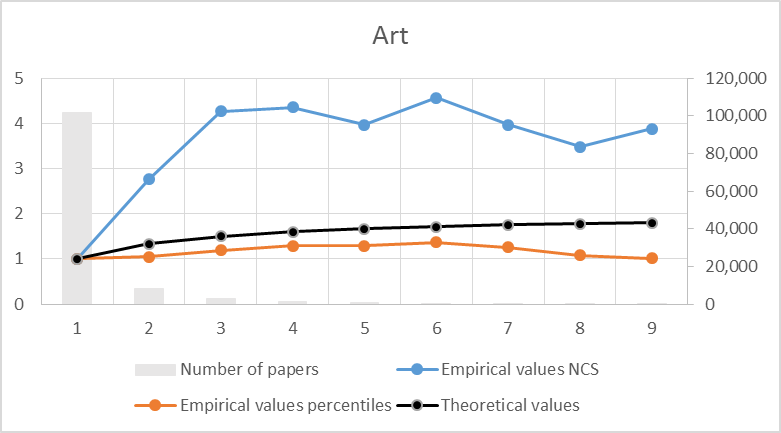}
    \label{figPH}
  \end{subfigure}\linebreak
\caption{{\protect\footnotesize Comparison of the theoretically derived
publication values for different numbers of authors with the empirically
derived field-normalized citation scores (NCS and percentiles) for 33
disciplines. The figures also show the number of papers with different
numbers of authors.}}
\end{figure}

\pagebreak

\section*{\protect\footnotesize References}

{\footnotesize 
\bibliographystyle{elsarticle-harv}
\bibliography{references}
}

\end{document}